\documentclass{eptcs}
\providecommand{\event}{TARK 2023}
\usepackage{amssymb,amsmath,amsthm,tikz}
\usetikzlibrary{calc}
\usepgflibrary{arrows}
\usepackage{caption,subcaption}
\usepackage{verbatim,mathdots}
\usepackage{rotating}
\usepackage{float}
\usepackage{array}
\usepackage{changepage} 
\usepackage{xspace}
\newtheorem{corollary}{Corollary}
\newtheorem{lemma}{Lemma}

\theoremstyle{definition}
\newtheorem{definition}{Definition}
\theoremstyle{remark}

\renewcommand{\phi}{\varphi}
\usepackage{todonotes}

\newcommand{\agents}{\mathcal{A}}

\newcommand{\K}{\textbf{K}\xspace}
\renewcommand{\phi}{\varphi}
\newcommand{\ag}{\mathcal{A}}
\newcommand{\prop}{\mathcal{P}}

\renewcommand{\theta}{\vartheta}

\title{Simple Axioms for Local Properties}
\def\titlerunning{Simple Axioms for Local Properties}
\def\authorrunning{Philippe Balbiani, Wiebe van der Hoek and Louwe B. Kuijer}
\date{}
\author{Philippe Balbiani\institute{Institut de Recherche en Informatique de Toulouse}\email{philippe.balbiani@irit.fr} \and Wiebe van der Hoek\institute{University of Liverpool}\email{wiebe@liverpool.ac.uk} \and Louwe B. Kuijer\institute{University of Liverpool}\email{lbkuijer@liverpool.ac.uk}}

\begin{document}
\maketitle

\begin{abstract}
Correspondence theory allows us to create sound and complete axiomatizations for modal logic on frames with certain properties. For example, if we restrict ourselves to transitive frames we should add the axiom $\square \phi \rightarrow \square\square\phi$ which, among other things, can be interpreted as positive introspection. One limitation of this technique is that the frame property and the axiom are assumed to hold globally, i.e., the relation is transitive throughout the frame, and the agent's knowledge satisfies positive introspection in every world.

In a modal logic with local properties, we can reason about properties that are not global. So, for example, transitivity might hold only in certain parts of the model and, as a result, the agent's knowledge might satisfy positive introspection in some worlds but not in others. Van Ditmarsch et al. \cite{ditmarsch_2012} introduced sound and complete axiomatizations for modal logics with certain local properties. Unfortunately, those axiomatizations are rather complex. Here, we introduce far simpler axiomatizations for a wide range of local properties.
\end{abstract}

\section{Introduction}
Modal logic is a formalism used throughout computer science, artificial intelligence and philosophy to represent various concepts including knowledge or belief (\emph{epistemic/doxastic logic}, e.g., \cite{meyer_1995, fagin_1995, ditmarsch_2015}), time (\emph{temporal logic}, e.g., \cite{pnueli_1977, benthem_1995}), necessity (\emph{alethic logic}, e.g., \cite{carnap_1947, blackburn_2007}), obligation (\emph{deontic logic}, e.g., \cite{aaqvist_1981, horty_2001}) and more. The main operator of modal logic is usually denoted $\square$, or, in a multi-agent setting, $\square_a$ where $a$ is an agent index. A formula $\square \phi$ can then be read, depending on the context, as ``$\phi$ is known'', ``$\phi$ is believed'', ``$\phi$ is true in every possible future'', ``$\phi$ is necessarily true'' or ``$\phi$ is obligatory''.

The most commonly used kind of semantics for modal logic uses \emph{relational models}, also known as \emph{Kripke models}. Such a model $M$ consists of three parts, $M=(W,R,V)$, where $W$ is a set of worlds or states, $R$ is an accessibility relation on those worlds (or, in the multi-agent case, a set of relations), and $V$ is a valuation that determines on which worlds an atomic formula is true. A formula $\square \phi$ then holds on a world $w_1$ if $\phi$ is true on every $w_2$ such that $(w_1,w_2)\in R$.

Depending on the specific concept being represented, some additional assumptions are typically made, however. Let us consider two such assumptions as examples. Firstly, in epistemic logic knowledge is generally assumed to be truthful, represented by the axiom $\square_a\phi \rightarrow \phi$, which can be read as ``if agent $a$ knows $\phi$, then $\phi$ is true''. Secondly, in alethic logic it is usually assumed that everything that is necessary is necessarily necessary, represented by the axiom $\square \phi\rightarrow \square\square \phi$.

Each such axiom \emph{corresponds}, in the precise technical sense of correspondence theory \cite{sahlqvist_1975, benthem_1976}, to a constraint on models. The axiom $\square_a\phi\rightarrow \phi$ corresponds to $a$'s accessibility relation being reflexive, $\neg \square \bot$ corresponds to the accessibility relation being serial and $\square \phi\rightarrow \square\square \phi$ corresponds to the relation being transitive.

Some non-standard assumptions can be handled in the same way. Suppose, for example, that agent $b$ is smarter and more well-informed than agent $a$, and that, as a result, everything known by $a$ is also known by $b$. This can be represented by an axiom $\square_a\phi\rightarrow \square_b\phi$ and the corresponding property that the relation $R(b)$ for agent $b$ is a subset of the relation $R(a)$ for agent $a$.

One limitation of this approach, to use axioms and their corresponding model conditions, is that the properties they represent are inherently global. If we assume the axiom $\square_a\phi\rightarrow \phi$, then $a$'s accessibility relation should be reflexive throughout the model. As such, not only is $a$'s knowledge truthful, that truthfulness must be common knowledge. Similarly, if we take $\square_a\phi\rightarrow \square_b\phi$ as an axiom then $b$ knows at least as much as $a$ in every world.

One solution to this issue was introduced in \cite{ditmarsch_2009, ditmarsch_2011, ditmarsch_2012} as \emph{local properties}. A local property is, as the name implies, a ``local'' variant of a property such as reflexivity, transitivity, or one relation being a subset of another. A world $w$ is locally reflexive if $w$ is a successor of itself, it is locally transitive if whenever $w_1$ is a successor of $w$ and $w_2$ is a successor of $w_1$, $w_2$ is also a successor of $w$, and $b$ locally knows more than $a$ in w if $\{w'\mid (w,w')\in R(b)\}\subseteq \{w'\mid (w,w')\in R(a)\}$. On the syntax side, for a given local property a special symbol $\theta$ is then introduced, that holds on a world if and only if that world satisfies the local property.

The main outcome of \cite{ditmarsch_2011, ditmarsch_2012} was a method to create sound and complete axiomatizations for logic with any local properties that satisfy a technical condition known as r-persistence. Unfortunately, while these axiomatizations are indeed sound and complete, they use a relatively complex introduction rule for local properties (see Section~\ref{sec:ditmarsch_approach} for details). Here, we introduce axiomatizations for many local properties that instead use an introduction axiom, which is also considerably simpler than the rule from \cite{ditmarsch_2012}. The downside of our approach is that it is less general than the one from \cite{ditmarsch_2012}, although it does include all commonly discussed local properties and we expect that it can be generalized.

The structure of this paper is as follows. First, in Section~\ref{sec:prelim} we introduce the necessary definitions and some background results. Then, in Section~\ref{sec:ditmarsch_approach}, we describe the axiomatizations from \cite{ditmarsch_2012}. Following that, in Section~\ref{sec:case}, we introduce a sound and complete axiomatization for one local property using a case study. Finally, in Section~\ref{sec:general}, we discuss axiomatizations for other local properties.

\section{Preliminaries}
\label{sec:prelim}
\subsection{Modal logic}
Before we define the specifics of local properties, it is useful to first define the usual semantics for modal logic.
Let $\ag$ be a finite set of agents and $\prop$ a countable set of propositional atoms.
\begin{definition}
A \emph{model} $M$ is a triple $M=(W,R,V)$ where $W$ is a set of worlds, $R: \ag\rightarrow 2^{W\times W}$ assigns to each agent an accessibility relation and $V:\prop\rightarrow 2^W$ is a valuation that assigns to each atom a subset of $W$.

We also write $w R_a w'$ for $(w,w')\in R(a)$ and denote the set $\{w'\mid w R_a w'\}$ by $R_a(w)$.

A pointed model is a pair $M,w$ where $M=(W,R,V)$ and $w\in W$.

A frame is a pair $F=(W,R)$, where $R: \ag\rightarrow 2^{W\times W}$.
\end{definition}

\begin{definition}
The modal formulas are given by the following normal form:
\begin{equation*}\phi ::= p \mid \top \mid \neg \phi \mid (\phi \vee \phi) \mid \square_a\phi,\end{equation*}
where $p\in \prop$ and $a\in \ag$.

The operators $\wedge, \rightarrow$ and $\leftrightarrow$ are defined as abbreviations in the usual way. Furthermore, we use $\lozenge_a$ as an abbreviation for $\neg\square_a\neg$. When $|\ag|=1$ we omit the index $a$, writing $\square$ and $\lozenge$ for $\square_a$ and $\lozenge_a$.
\end{definition}
The semantics are as usual.
\begin{definition}
Let $M,w$ be a pointed model. The satisfaction relation $\models$ is given recursively by
$$\begin{array}{lcl}
M,w\models p & \Leftrightarrow & w\in V(p)\\
M,w\models \neg \phi & \Leftrightarrow & M,w\not\models \phi\\
M,w\models \phi_1\vee \phi_2 & \Leftrightarrow & M,w\models \phi_1 \text{ or } M,w\models \phi_2\\
M,w\models \square_a\phi & \Leftrightarrow & \forall w'\in R_a(w): M,w'\models \phi
\end{array}$$
If $M,w\models \phi$ for every $w\in W$, we write $M\models \phi$. If $M\models \phi$ for every $M$ we say that $\phi$ is \emph{valid} and write $\models \phi$.

If $F=(W,R)$ is a frame, we say that $F,w\models \phi$ if $(W,R,V),w\models \phi$ for every $V$.
\end{definition}
The usual proof system \K is as follows.
\begin{definition}
The proof system \K is given by the following axioms and rules.
\begin{center}
\begin{tabular}{ll}
T & any substitution instance of a validity of propositional logic\\
K & $\square_a (\phi\rightarrow \psi)\rightarrow (\square_a\phi\rightarrow \square_a\psi)$\\
Nec & From $\psi$, infer $\square_a\psi$\\
MP & From $\phi \rightarrow \psi$ and $\phi$, infer $\psi$.
\end{tabular}
\end{center}
\end{definition}
It is well known that \K is sound and strongly complete for modal logic. All axiomatizations that we discuss in this paper extend \K.

Finally, we should define bisimulations, as these will become important later.
\begin{definition}
Let $M_1=(W_1,R_1,V_1)$ and $M_2=(W_2,R_2,V_2)$ be models. A \emph{bisimulation} is a relation $\sim \subseteq W_1\times W_2$ such that the following three properties hold.
\begin{description}
	\item[Atomic agreement] If $w_1\sim w_2$ then for all $p\in \prop$, $w_1\in V(p)$ iff $w_2\in V(p)$.
	\item[Forth] If $w_1\sim w_2$ and $(w_1,w_1')\in R_1(a)$, then there is a $w_2'\in W_2$ such that $(w_2,w_2')\in R_2(a)$ and $w_1'\sim w_2'$.
	\item[Back] If $w_1\sim w_2$ and $(w_2,w_2')\in R_2(a)$, then there is a $w_1'\in W_1$ such that $(w_1,w_1')\in R_1(a)$ and $w_1'\sim w_2'$.
\end{description}
\end{definition}
Famously, modal logic is the bisimulation-invariant fragment of first-order logic \cite{benthem_1976}. In particular, modal logic is invariant under bisimulation, so if $w_1\sim w_2$ then the two worlds satisfy the same modal formulas.

\subsection{Local properties}
\newcommand{\Thetaser}{\Theta_\mathit{ser}}
\newcommand{\thetaser}{\theta_\mathit{ser}}
Now we can define local properties, and their interaction with modal logic.
\begin{definition}
Let $X$ be a set of first order variables. A \emph{local property} is a formula with one free variable in the first order language given by 
\begin{equation*}\Theta ::= (x, y)\in R(a) \mid x=y\mid \neg \Theta \mid (\Theta \vee \Theta) \mid \forall x\Theta,\end{equation*}
where $a\in \ag$ and $x,y\in X$.
\end{definition}
We follow \cite{ditmarsch_2012} in this definition, by allowing only relational predicates, $(x,y)\in R(a)$, and not valuation predicates $x\in V(p)$. This restriction is not necessary for our analysis, but it seems harmless, given that we do not know of any interesting local properties that depend on such valuation predicates.
A local property is a first order formula, so it can be evaluated on models, where we take the set $W$ of worlds to be the domain of quantification. Furthermore, because we restrict ourselves to relational predicates, the valuation does not affect the value of any local property, so we can also evaluate a local property on the frame underlying a model instead of on the model itself.

We will focus primarily on local properties that are not invariant under bisimulation. This is not because formulas that are invariant under bisimulation are necessarily uninteresting, but because any such property is equivalent to a formula of modal logic. Consider, for example, the property of local seriality, $\Thetaser = \exists x (w,x)\in R$. We have $\Thetaser(w)$ if and only if $M,w\models \lozenge \top$, so if we wish to reason about local seriality we need not introduce a special symbol $\thetaser$ but can instead reason about $\lozenge \top$.

If we wish to reason about a given local property $\Theta$ in modal logic, we add an extra symbol $\theta$ to the language of modal logic. On a technical level, $\theta$ can be seen as either a nullary modality or a designated propositional atom. In order to emphasize its special role, we denote its extension separately in our notation for models. Specifically, we write $M=(W,R,\Delta,V)$, where $\Delta\subseteq W$ and $M,w\models \theta$ iff $w\in \Delta$. If we look at multiple local properties at a time, each is represented by a different symbol, i.e., $M=(W,R,\Delta_1,\cdots, \Delta_k, V)$, with $M,w\models \theta_i$ iff $w\in \Delta_i$. Because $\Delta$ is simply a subset of $W$, there is no inherent guarantee that the atom $\theta$ is in any way related to the property $\Theta$. In order to force this connection, we look at models that are in \emph{harmony}.

\begin{definition}
A model $M=(W,R,\Delta,V)$ is in \emph{$\Theta$-$\theta$-harmony} if for every $w\in W$, we have $w\in \Delta$ if and only if $\Theta(w)$.
\end{definition}
When $\Theta$ and $\theta$ are clear from context we omit reference to them, and simply say that a model is in harmony.
Our goal, and the goal of \cite{ditmarsch_2011, ditmarsch_2012}, is to find axiomatizations that are sound and complete for the class of models that are in harmony.

Finally, we need the notion of a modal formula locally defining a first-order property.
\begin{definition}
Let $\phi$ be a schema of modal logic and $\Theta$ a local property. The formula $\phi$ \emph{locally defines} $\Theta$ if for every frame $F=(W,R)$ and every $w\in W$, we have $F,w\models \phi$ if and only if $\Theta(w)$.
\end{definition}
Generally, the schema that we will use to locally define a property is the same one that globally corresponds to that property. So, for example, the property of (local) transitivity is locally defined by $\square \psi\rightarrow \square\square \psi$, and (local) reflexivity is locally defined by $\square\psi\rightarrow \psi$. 
If $\varphi(p_{1},\ldots,p_{n})$ is a modal formula constructed from the propositional atoms $p_{1},\ldots,p_{n}$ then for all modal formulas $\chi_{1},\ldots,\chi_{n}$, $\varphi(\chi_{1},\ldots,\chi_{n})$ will denote the modal formula obtained from $\varphi(p_{1},\ldots,p_{n})$ by respectively replacing the occurrences of $p_{1},\ldots,p_{n}$ by $\chi_{1},\ldots,\chi_{n}$.

\section{Axiomatizations using an inference rule}
\label{sec:ditmarsch_approach}
In order for the method from \cite{ditmarsch_2012} to apply for a local property $\Theta$, we require two things. Firstly, we must have a schema $\phi$ that locally defines $\Theta$. Secondly, $\phi$ should be locally r-persistent. Defining this property requires a lot of further definitions, so for practical reasons we will not give such a definition here and instead refer the reader to \cite{goranko_1998} for details. We will note, however, that the axioms corresponding to the usual properties (transitivity, reflexivity, etc.) are locally r-persistent.

If these conditions are satisfied, the axiomatizations from \cite{ditmarsch_2012} work by adding an \emph{elimination axiom} and an \emph{introduction rule} for $\theta$ to the proof system \K. The elimination axiom E$\theta$  is quite simple.
\begin{center}\begin{tabular}{ll} E$\theta$ & \hspace{15pt}$\theta\rightarrow\phi$,\end{tabular}\end{center}
\newcommand{\Thetaeuc}{\Theta_\mathit{Euc}}
\newcommand{\thetaeuc}{\theta_\mathit{Euc}}
where $\phi$ is the schema that locally defines $\Theta$. Consider, for example, local Euclidicity, formally given by $\Thetaeuc = \forall x,y(((w,x)\in R(a) \wedge (w,y)\in R(a))\rightarrow (x,y)\in R(a))$. This property is locally defined by the schema $\lozenge_a \psi\rightarrow \square_a\lozenge_a \psi$, i.e., negative introspection. Hence the elimination axiom E$\thetaeuc$ is given by $\thetaeuc\rightarrow (\lozenge_a \psi\rightarrow \square_a\lozenge_a \psi)$. This, of course, is exactly as desired, since it means that in every world where $\thetaeuc$ holds, the agent $a$ is capable of negative introspection.

The introduction rule is more complex, and before we can formally state it we first need to introduce \emph{pseudo-modalities}. Let $s = x_1,\cdots,x_n$ be a (possibly empty) finite sequence of pairs $x_i = (a_i,\chi_i)$, where $a_i\in \agents$ and $\chi_i$ is a formula of modal logic. The pseudo-modality $[s]$ is then an abbreviation, where $[s] \psi$ stands for $\chi_1\rightarrow \square_{a_1}(\chi_2\rightarrow \square_{a_2}(\chi_3 \rightarrow \square_{a_3}(\cdots \square_{a_{n-1}}(\chi_n \rightarrow \square_{a_n} \psi))$.

Let $k$ be the number of different schematic variables in $\phi$. Then the introduction rule I$\theta$ is as follows.
\begin{center}\begin{tabular}{ll}I$\theta$ & \hspace{15pt} from $[s]\phi(p_1,\cdots,p_k)$, infer $[s]\theta$,\end{tabular}\end{center}
where $[s]$ is a pseudo-modality and $p_1,\cdots,p_k$ are fresh atoms.

Consider again the property $\Thetaeuc$. By taking $s$ to be the empty sequence, the rule I$\thetaeuc$ allows us to infer $\thetaeuc$ from $\lozenge_a p \rightarrow \square_a\lozenge_a p$. By taking a non-empty sequence $s$, for example $s= (\psi_1, b), (\psi_2,a), (\psi_3,c)$, it also allows us to infer
\[\psi_1\rightarrow \square_b (\psi_2\rightarrow \square_a (\psi_3\rightarrow \square_c \thetaeuc))\]
from
\[\psi_1\rightarrow \square_b (\psi_2\rightarrow \square_a (\psi_3\rightarrow \square_c (\lozenge_a p \rightarrow \square_a\lozenge_ap))),\]
as long as $p$ does not occur in $\psi_1, \psi_2$ and $\psi_3$.

The main result from \cite{ditmarsch_2012} is that if $\Theta_1,\cdots, \Theta_n$ are local properties that satisfy the required conditions, then the proof system $\K + \text{E}\theta_1 + \text{I}\theta_1 + \cdots + \text{E}\theta_n + \text{I}\theta_n$ is sound and complete for the class of models that are in  $\Theta_i$-$\theta_i$-harmony for all $1\leq i \leq n$.

\section{Case study: transitivity}
\label{sec:case}
\newcommand{\Thetatrans}{\Theta_\mathit{tr}}
\newcommand{\thetatrans}{\theta_\mathit{tr}}
Here, we introduce alternative axiomatizations for local properties. We use the same elimination axiom, but instead of an introduction rule we use an introduction axiom. Furthermore, our introduction axiom is syntactically simpler than the rules from \cite{ditmarsch_2012}. Our axiomatization is based on the observation that most of the local properties that seem interesting (including, but not limited to, the examples from \cite{ditmarsch_2009, ditmarsch_2011, ditmarsch_2012}) are not merely \emph{not preserved} under bisimilarity, but \emph{anti-preserved}, in the sense that apart from some trivial exceptions,\footnote{These exceptions apply when a world has no successors. For example, if $\forall x \neg (w,x)\in R$, then $w \sim w'$ implies that $w'$ is locally Euclidean, as well as locally transitive, locally dense, et cetera.} for every $M$ and $w$ such that $\Theta(w)$, there are $M',w'$ such that $w\sim w'$ and $\neg\Theta(w')$.

In this section we use the property of local transitivity to show how we can leverage this anti-preservation in order to obtain a sound and strongly complete axiomatization. So take \[\Thetatrans = \forall x,y (((w,x)\in R \wedge (x,y)\in R)\rightarrow (w,y)\in R).\]
Our elimination axiom for $\thetatrans$ is the same as that of \cite{ditmarsch_2012}.
\begin{center}\begin{tabular}{ll}
E$\thetatrans$ & \hspace{15pt}$\thetatrans\rightarrow (\square \psi\rightarrow \square\square\psi).$
\end{tabular}\end{center}
Furthermore, it is clear that when the antecedent of the implication in $\Thetatrans$ is not satisfied, $\Thetatrans$ trivially holds. So if $\forall x,y \neg ((w,x)\in R \wedge (x,y)\in R)$, then $\Thetatrans(w)$. Unlike $\Thetatrans$ itself, this latter property is invariant under bisimulation. In fact, it is equivalent to the modal formula $\square\square\bot$. The following \emph{trivial introduction axiom} is therefore sound.
\begin{center}\begin{tabular}{ll}TI$\thetatrans$ & \hspace{15pt} $\square\square\bot\rightarrow \thetatrans$\end{tabular}\end{center}
We will show that E$\thetatrans$ and TI$\thetatrans$ suffice, so $\K + \text{E}\thetatrans + \text{TI}\thetatrans$ is sound and complete for the class of models that are in $\Thetatrans$-$\thetatrans$-harmony. For an informal overview of why we do not need any further introduction axioms, note that for every pointed model $M,w$, if $M,w\not \models \square\square \bot$ and $\Thetatrans(w)$ then there is a bisimilar model $M',w'$ such that $\neg\Thetatrans(w')$. In particular, the tree unravelling of $M$ will have that property.

We will show that, as a consequence, no introduction axiom $\chi\rightarrow \thetatrans$ can be sound unless $\chi$ contains $\thetatrans$ (which would make it a rather useless introduction axiom) or $\chi$ implies $\square\square\bot$ (which would render the axiom superfluous in the presence of TI$\thetatrans$). So suppose towards a contradiction that there is some modal formula $\chi$ such that (i) $\thetatrans$ does not occur in $\chi$, (ii) $\chi$ does not imply $\square\square\bot$ and (iii) $\chi\rightarrow \thetatrans$ is valid on models that are in harmony. Then there is some $M,w$ such that $M,w\not\models\square\square\bot$ and $M,w\models \chi$. Take $M',w'$ to be the bisimilar model such that $\neg \Thetatrans(w')$. Because $\chi$ is a modal formula, it is invariant under bisimulation. Hence $M',w'\models \chi$. Furthermore, because $\chi$ does not contain $\thetatrans$, we can choose $M'$ to be in harmony. Since $\neg \Thetatrans(w')$ holds this implies that $M',w'\not\models \thetatrans$, which contradicts the soundness of $\chi\rightarrow \thetatrans$ on models in harmony.

From this contradiction, it follows that no axiom $\chi\rightarrow \thetatrans$ can be sound unless $\chi$ contains the symbol $\thetatrans$ or $\chi$ implies $\square\square\bot$. This does not fully suffice to prove that $\K+\text{E}\thetatrans+\text{TI}\thetatrans$ is complete for models in harmony, but it does show why we should not expect to need any further introduction axioms.

In order to turn this informal proof sketch into a full proof, let us first introduce one further auxiliary definition.
\begin{definition}
A model $M$ is \emph{$\Thetatrans$-$\thetatrans$-nice} if
\begin{enumerate}
	\item for every $w$, if $M,w\models \thetatrans$ then $\Thetatrans(w)$ and
	\item if $M,w\models \square\square\bot$ then $M,w\models\thetatrans$.
\end{enumerate}
\end{definition}
Every model that is in harmony is also nice, but not necessarily vice versa. It is also quite easy to see that $\K + \text{E}\thetatrans + \text{TI}\thetatrans$ is sound and strongly complete for the class of nice models, since E$\thetatrans$ directly corresponds to the first condition of niceness and TI$\thetatrans$ corresponds to the second condition. This completeness can be proven using the standard canonical model construction. We will show that for every nice model, there is a bisimilar harmonious model.

\begin{definition}
Let $M=(W,R,\Delta,V)$ be a nice model. We define a sequence of models as follows:
\begin{itemize}
	\item $M_0=(W',R_0,\Delta',V')$ is the tree unravelling of $M$, i.e., 
    \begin{itemize}
        \item $W'$ is the set of finite sequences $w' = (w_1,\cdots, w_n)$ such that (i) $w_j\in W$ for all $1\leq j \leq n$ and (ii) $(w_j,w_j+1)\in R$ for all $1\leq j < n$,
        \item $(w_1',w_2')\in R_0$ if and only if $w_1' = (w_1,\cdots, w_n)$ and $w_2' = (w_1,\cdots, w_n, w_{n+1})$ for some $w_1,\cdots, \allowbreak w_{n+1}\in W$,
        \item $(w_1,\cdots, w_n)\in \Delta'$ if and only if $w_n\in \Delta$,
        \item $(w_1,\cdots, w_n) \in V'(p)$ if and only if $w_n\in V(p)$.
    \end{itemize}
	\item $M_{i+i}= (W',R_{i+1},\Delta',V')$, where $(w_1',w_2')\in R_{i+1}$ if and only if 
	\begin{itemize}
		\item $(w_1',w_2')\in R_i$ or
		\item $w_1'\in \Delta'$ and there is a $w_3'$ such that $(w_1',w_3')\in R_i$ and $(w_3',w_2')\in R_i$.
	\end{itemize}
	\item $M_\infty = (W',R_\infty,\Delta',V')$, where $R_\infty=\bigcup_{i\in \mathbb{N}}R_i$.
\end{itemize}
\end{definition}
Because $M_0$ is a tree model, no world is locally transitive (unless none of its successors have a successor). We then add additional edges to the relation, but only where needed to make every $w'\in \Delta'$ locally transitive. As a consequence, we will be able to show that there is a total bisimulation between $M$ and $M_\infty$, and that $M_\infty$ is in harmony.

Note that all models $M_i$, for $i\in \mathbb{N}\cup\{\infty\}$ use the same set $W'$ of worlds, and that this world was obtained by taking the tree unravelling of $M$. As such, every $w'=(w_1,\cdots,w_n)\in W'$ has a unique original, namely $w_n$, in $W$.

\begin{lemma}
Take any $x',y'\in W'$ and let $x,y\in W$ be the originals of $x'$ and $y'$, respectively. If $(x',y')\in R_i$ for some $i\in \mathbb{N}\cup \{\infty\}$, then $(x,y)\in R$.
\end{lemma}
\begin{proof}
By induction on $i$. As base case, suppose that $i=0$. Then the lemma follows immediately from the fact that $M_0$ is the tree unravelling of $M$. Suppose then, as induction hypothesis, that $i\in \mathbb{N}_{>0}$ and that the lemma holds for all $i'<i$.

Take any $(x',y')\in R_i$. If already $(x', y')\in R_{i-1}$, then by the induction hypothesis we have $(x, y)\in R$, as was to be shown. 

If $(x',y')\in R_i\setminus R_{i-1}$, then the arrow must have been added in stage $i$, meaning that $x'\in \Delta'$ and there is some $z'$ such that $(x',z')\in R_{i-1}$ and $(z',y')\in R_{i-1}$. Then, by the induction hypothesis, $(x,z)\in R$ and $(z,y)\in R$, where $z$ is the original of $z'$. Furthermore, because $M_0$ is the tree unravelling of $M$, and all $M_i$ use the same set $\Delta'$, it follows from $x'\in \Delta'$ that $x\in \Delta$. 

From the fact that $M$ is nice, it then follows that $\Thetatrans(w)$, so $w$ is locally transitive. Because $(x,z)\in R$ and $(z,y)\in R$, we then have $(x,y)\in R$, as was to be shown.
This completes the induction step for $i\in \mathbb{N}_{>0}$. Finally, if $(x',y')\in R_\infty$, then $(x',y')\in R_j$ for some $j\in \mathbb{N}$, and therefore $(x,y)\in R$.
\end{proof}

Now, we can take the important step of showing bisimilarity.
\begin{lemma}
Let $\sim \subseteq W\times W'$ be the relation such that $x\sim y'$ if and only if $x$ is the original of $y'$. Then $\sim$ is a bisimulation between $M$ and $M_i$, for every $i\in \mathbb{N}\cup \{\infty\}$.
\end{lemma}
\begin{proof}
We show that $\sim$ satisfies the three conditions of bisimulation. Take any $w_1\in W$ and $w_1'\in W'$ such that $w_1\sim w_1'$.
\begin{description}
	\item[Atoms] Because $M_0$ is the tree unraveling of $M$, and $M_i$ and $M'$ only differ from $M_0$ by the addition of edges, we have $w_1\in V(p)$ iff $w_1'\in V'(p)$.
	\item[Forth] Take any $w_2$ such that $(w_1,w_2)\in R$. As $M_0$ is the tree unraveling of $M$, there is some $w_2'\in W'$ such that $(w_1',w_2')\in R_0$ and $w_2$ is the original of $w_2'$ (and therefore $w_2\sim w_2'$). Furthermore, $(w_1',w_2')\in R_0$ implies $(w_1',w_2')\in R_i$ for every $i\in\mathbb{N}\cup\{\infty\}$.
	\item[Back] Take any $w_2'$ such that $(w_1',w_2')\in R_i$ for some $i\in \mathbb{N}\cup\{\infty\}$. Then by the preceding lemma
	we have $(w_1,w_2)\in R$, where $w_2$ is the original of $w_2'$. Hence we also have $w_2\sim w_2'$.
\end{description}
\end{proof}

Left to show is that $M_\infty$ is $\Thetatrans$-$\thetatrans$-harmonious.
\begin{lemma}
If $x'\in \Delta'$, $(x',y')\in R_\infty$ and $(y',z')\in R_\infty$ then $(x',z')\in R_\infty$.
\end{lemma}
\begin{proof}
If $(x',y')\in R_\infty$ and $(y',z')\in R_\infty$, then there is some $i\in \mathbb{N}$ such that $(x',y')\in R_i$ and $(y',z')\in R_i$. As $x'\in \Delta'$, this implies that $(x',z')\in R_{i+1}$, and therefore $(x',z')\in R_\infty$.
\end{proof}
\begin{lemma}
If $x'\not \in \Delta'$, then there are $y',z'$ such that $(x',y')\in R_\infty$ and $(y',z')\in R_\infty$ while $(x',z')\not\in R_\infty$.
\end{lemma}
\begin{proof}
From $x'\not \in \Delta'$ it follows that $x\not\in \Delta$. Because $M$ is nice, this implies that there are $y,z$ such that $(x, y)\in R$ and $(y,z)\in R$. Then there are $y',z'\in W'$ such that $(x',y')\in R_0$ and $(y',z')\in R_0$. As $M_0$ is a tree model, we have $(x',z')\not\in R_0$. Furthermore, because $x'\not\in \Delta'$, no edge from $x'$ to $z'$ is added in any model $M_i$. Hence $(x',z')\not\in R_\infty$.
\end{proof}
\begin{corollary}
The model $M'=M_\infty$ is $\Thetatrans$-$\thetatrans$-harmonious.
\end{corollary}

It now follows quite easily that $\K+\text{E}\thetatrans+\text{TI}\thetatrans$ is sound and complete.
\begin{corollary}
The axiomatization $\K+\text{E}\thetatrans+\text{TI}\thetatrans$ is sound and strongly complete for the class of models that are $\Thetatrans$-$\thetatrans$-harmonious.
\end{corollary}
\begin{proof}
The axiomatization is sound and strongly complete for $\Thetatrans$-$\thetatrans$-nice models. Furthermore, every nice model $M$ can be transformed into a bisimilar model $M_\infty$ that is in harmony. It follows that the axiomatization is sound and strongly complete for harmonious models.
\end{proof}

\section{Axioms for local properties}
\label{sec:general}

\newcommand{\Thetaref}{\Theta_\mathit{ref}}
\newcommand{\thetaref}{\theta_\mathit{ref}}
\newcommand{\Thetasym}{\Theta_\mathit{sym}}
\newcommand{\thetasym}{\theta_\mathit{sym}}
\newcommand{\Thetasup}{\Theta_\mathit{sup(a,b)}}
\newcommand{\thetasup}{\theta_\mathit{sup(a,b)}}
\newcommand{\Thetadens}{\Theta_\mathit{dense}}
\newcommand{\thetadens}{\theta_\mathit{dense}}
\newcommand{\Thetafunc}{\Theta_\mathit{func}}
\newcommand{\thetafunc}{\theta_\mathit{func}}

Many other local properties can be given a simple axiomatization in a way similar to local transitivity. Here, we consider local reflexivity, Euclidicity, symmetry, superset, density and functionality. This includes all of the examples from \cite{ditmarsch_2012}. Note that each of these properties requires the existence of one or more edges. In the case of reflexivity this requirement is unconditional; in order for $w$ to be locally reflexive there must be an edge $(w,w)\in R$. For the other properties, the requirement for the edge to exist is conditional on the existence of one or more other edges. For example, local Euclidicity requires that if $(w_1,w_2)\in R$ and $(w_1,w_3)\in R$ then $(w_2,w_3)\in R$, while local symmetry requires that if $(w_1,w_2)\in R$ then $(w_2,w_1)\in R$. It is important to note that for each of these properties it is an edge between two specific worlds that needs to exist. In the example of local Euclidicity, there must be an edge from the world $w_2$ to the world $w_3$. Any other edge, even if it is from $w_2$ to a world that is bisimilar to $w_3$, does not suffice.

An introduction axiom takes the form $\chi\rightarrow \theta$, where $\chi$ is a modal formula that does not contain $\theta$. In order for this axiom to be sound for the models where $\Theta$ and $\theta$ are in harmony, it therefore has to be the case that whenever $\chi$ is true, either (i) the condition under which the local property $\Theta$ requires an edge to exist is false or (ii) the specific edge(s) required by $\Theta$ do exist.

There is no modal formula that implies the existence of any specific edge, so $\chi$ cannot guarantee that option (ii) is the case.\footnote{There are modal formulas, such as $\lozenge \top$, that guarantee the existence of an edge, but these formulas don't guarantee the existence of a \emph{specific} edge.} There are modal formulas that imply the \emph{non}existence of a particular edge, however. For example, if $\square\bot$ is true in $w_1$ then there are no edges that start in $w_1$ and therefore, in particular, it is not the case that $(w_1,w_2)\in R$ and $(w_1,w_3)\in R$. For most local properties, it is therefore possible for $\chi$ to guarantee that condition (i) holds, and therefore for the introduction axiom to be sound, if we take $\chi = \square \bot$ or $\chi = \square\square \bot$. The unconditional nature of reflexivity makes it an exception in this regard, $\chi \rightarrow \thetaref$ is sound only if $\chi$ is unsatisfiable.

Of course the above reasoning only tells us when an introduction axiom of this form is sound, completeness remains to be shown. As in the preceding section, showing completeness is done by proving that every nice model can be transformed into a bisimilar harmonious one.

Let us now consider each of the aforementioned local properties in some more detail.
Local reflexivity, given by $\Thetaref = (w,w)\in R$, has an even simpler axiomatization than local transitivity. The elimination axiom is as one would expect, E$\thetaref$ is $\thetaref\rightarrow (\square\phi\rightarrow \phi)$. But unlike $\thetatrans$, we do not require any introduction axiom for $\thetaref$. This is because, unlike any of the other properties that we consider here, $\Thetaref$ does not contain an implication of which the antecedent can be false. As a result, there is no modal formula $\phi$ such that $M,w\models \phi$ trivially implies that $\Thetaref(w)$.

Nice models with respect to $\Thetaref$ and $\thetaref$ are therefore simply the ones where $\Thetaref(w)$ holds for all $w\in \Delta$. As in the case of transitivity, we can then unravel any nice model $M$ into a tree model $M_0$, and then modify that tree model into a bisimilar model that is in harmony. It follows that $\K + \text{E}\thetaref$ is sound and strongly complete for the class of models that are in $\Thetaref$-$\thetaref$-harmony.

Local Euclidity, symmetry and density are given by 
\[\Thetaeuc = \forall x,y (((w,x)\in R \wedge (w,y)\in R)\rightarrow (x,y)\in R),\]
\[\Thetasym = \forall x ((w,x)\in R\rightarrow (x,w)\in R)\]
and
\[\Thetadens = \forall x \exists y ((w,x)\in R\rightarrow ((w,y)\in R)\wedge (y,x)\in R),\]
respectively. The elimination axioms for these properties are as one would expect:
\begin{center}\begin{tabular}{ll}E$\thetaeuc$ & \hspace{15pt} $\thetaeuc\rightarrow (\lozenge \phi\rightarrow \square\lozenge \phi)$\\
E$\thetasym$ & \hspace{15pt} $\thetasym\rightarrow (\phi\rightarrow \square\lozenge \phi)$\\
E$\thetadens$ & \hspace{15pt} $\thetadens\rightarrow (\lozenge \phi\rightarrow\lozenge\lozenge\phi)$\end{tabular}\end{center}
For each of these properties, the antecedent of the implication is trivially false when there is no $x$ such that $(w,x)\in R$. Hence the trivial introduction axioms are simply $\square\bot\rightarrow \thetaeuc$, $\square\bot\rightarrow\thetasym$ and $\square\bot\rightarrow \thetadens$. Completeness is shown as before, by turning nice models into tree models and then those tree models into harmonious models. It follows that $\K + \text{E}\thetaeuc + \text{TI}\thetaeuc$ is sound and strongly complete for models that are in $\Thetaeuc$-$\thetaeuc$-harmony, $\K + \text{E}\thetasym + \text{TI}\thetasym$ is sound and complete for models in $\Thetasym$-$\thetasym$-harmony and $\K + \text{E}\thetadens + \text{TI}\thetadens$ is sound and complete for models in $\Thetadens$-$\thetadens$-harmony.

The local property of $R(a)$ being a superset of $R(b)$ is given by $\Thetasup = \forall x ((w,x)\in R(b)\rightarrow (w,x)\in R(a))$. Recall that this property can be read as $b$ knowing at least as much as $a$. The elimination axiom E$\thetasup$ is therefore, unsurprisingly, $\thetasup\rightarrow (\square_a\phi\rightarrow \square_b\phi)$. The antecedent of $\Thetasup$ is false when there is no $x$ such that $(w,x)\in R(b)$, so when $M,w\models \square_b\bot$. The introduction axiom TI$\thetasup$ is therefore $\square_b\bot\rightarrow \thetasup$. The presence of the different agents $a$ and $b$ does not interfere in our procedure of unraveling $M$ and turning that unraveling into a harmonious model, so $\K + \text{E}\thetasup + \text{TI}\thetasup$ is sound and complete for models in $\Thetasup$-$\thetasup$-harmony.

Finally, let us consider local functionality, $\Thetafunc = \forall x,y ((R(w,x)\wedge R(w,y))\rightarrow x=y)$. The elimination axiom E$\thetafunc$ is $\thetafunc\rightarrow ((\lozenge \phi\wedge \lozenge \psi)\rightarrow \lozenge (\phi\wedge\psi))$, and the trivial elimination axiom TI$\thetafunc$ by $\square\bot\rightarrow \thetafunc$. As with the other properties we can then unravel any nice model $M$ into a tree model $M_0$, and turn $M_0$ into a harmonious model $M_\infty$. The only difference is that in this case it does not suffice to merely add edges in order to obtain $M_{i+1}$ from $M_i$. Instead, for every world $w'\not\in \Delta'$, if there is only one $x'$ such that $(w',x')\in R_i$ then we need to create a copy of the sub-tree rooted in $x'$, the root of which we will call $x''$. Then we add this tree to $M_{i+1}$, as well as an edge $(w,x'')\in R_{i+1}$. The result of this procedure will be a harmonious model $M_\infty$ such that $M_\infty,w'$ is bisimilar to $M,w$. It follows that $\K + \text{E}\thetafunc + \text{TI}\thetafunc$ is sound and complete for models in $\Thetafunc$-$\thetafunc$-harmony.

In summary, we have the following elimination and introduction axioms.
\begin{center}\begin{tabular}{l|l|l}
Property & Elimination axiom & Introduction axiom \\\hline
$\Thetatrans$ & $\thetatrans\rightarrow (\square\phi\rightarrow\square\square\phi)$ & $\square\square\bot\rightarrow \thetatrans$\\
$\Thetaref$ & $\thetaref\rightarrow (\square\phi\rightarrow \phi)$ & -\\
$\Thetaeuc$ & $\thetaeuc \rightarrow (\lozenge \phi\rightarrow \square\lozenge \phi)$ & $\square\bot\rightarrow \thetaeuc$\\
$\Thetasym$ & $\thetasym\rightarrow (\phi\rightarrow \square\lozenge \phi)$ & $\square \bot \rightarrow \thetasym$\\
$\Thetasup$ & $\thetasup\rightarrow (\square_a\phi\rightarrow \square_b\phi)$ & $\square_b\bot\rightarrow \thetasup$\\
$\Thetadens$ & $\thetadens\rightarrow (\lozenge \phi\rightarrow \lozenge\lozenge\phi)$ & $\square\bot \rightarrow \thetadens$\\
$\Thetafunc$ & $\thetafunc\rightarrow ((\lozenge\phi\wedge\lozenge\psi)\rightarrow\lozenge(\phi\wedge\psi))$ & $\square\bot\rightarrow \thetafunc$
\end{tabular}\end{center}

Unfortunately, while this method does yield simple and elegant axiomatizations for these local properties, it is not as general as the method from \cite{ditmarsch_2012}. For one thing, we do not have a general technique that allows for the automatic generation of axiomatizations for large classes of local properties. Furthermore, unlike \cite{ditmarsch_2012} adding the axioms for multiple local properties does not necessarily yield a sound and complete axiomatization for the class of models that are in harmony for each property.

For example, $\K + \text{E}\thetaeuc + \text{TI}\thetaeuc + \text{E}\thetaref$ is not sound and complete for the models that are in both $\Thetaeuc$-$\thetaeuc$ and $\Thetaref$-$\thetaref$-harmony. This is because if $\Thetaeuc(w_1)$ and $(w_1,w_2)\in R$ then we also have $\Thetaref(w_2)$. Hence we would need the additional introduction axiom I$\thetaeuc\thetaref$, namely $\thetaeuc \rightarrow \square\thetaref$. With that additional axiom, we can once again use the same construction to turn every nice model into a harmonious model, so $\K + \text{E}\thetaeuc + \text{TI}\thetaeuc + \text{E}\thetaref + \text{I}\thetaeuc\thetaref$ is sound and complete for the models in $\Thetaeuc$-$\thetaeuc$ and $\Thetaref$-$\thetaref$-harmony.

Such additional axioms are not always needed. For example, $\K + \text{E}\thetatrans + \text{TI}\thetatrans + \text{E}\thetaref$ is sound and complete for the models in $\Thetatrans$-$\thetatrans$ and $\Thetaref$-$\thetaref$-harmony. Furthermore, where needed the extra axioms seems relatively easy to find. Yet we do not currently have a systematic way of determining whether an additional axiom is required and, if so, which axiom.

\section{Conclusion}
We have presented an alternative way to create axiomatizations for local properties. In contrast to the existing axiomatizations from \cite{ditmarsch_2012}, our approach uses introduction \emph{axioms}, as opposed to introduction \emph{rules}. Furthermore, our axioms are simpler and, in our opinion, more elegant than the rules from \cite{ditmarsch_2012}.

The price we pay for this simplicity is that we do not, as of yet, have a general way to create axiomatizations for further local properties, or for combinations of multiple local properties. Given the extremely strong similarities between the completeness proofs for the axiomatizations that we considered, we expect that some kind of generalization is possible, but we have not found it yet.

As such, the main direction for further work would be to find such generalizations.

\bibliographystyle{eptcs}
\bibliography{local_properties}

\end{document}